\documentclass[journal,10pt,english,twocolumn]{IEEEtran}
\usepackage[T1]{fontenc}
\usepackage[latin9]{inputenc}
\usepackage{babel}
\usepackage{epsfig}
\usepackage{latexsym}
\usepackage{color}
\usepackage{epsf}
\usepackage{subfig}
\usepackage{amsthm}
\usepackage{amsmath}
\usepackage{graphicx}
\usepackage{amssymb}
\usepackage{psfrag}
\theoremstyle{plain}
 \theoremstyle{plain}
\newtheorem{thm}{Theorem}
 \theoremstyle{definition}
  
  \theoremstyle{ }
  \newtheorem{defn}[thm]{Definition}
  \theoremstyle{plain}
  \newtheorem*{cor*}{Corollary}
  \theoremstyle{plain}
  \newtheorem{lem}{Lemma}
\newtheorem{theorem}{Theorem}
\newcommand{\argmax}{\operatornamewithlimits{argmax}}
\begin{document}
\title{A Greedy Link Scheduler for Wireless Networks with Fading Channels}
\date{}
\author{Arun ~Sridharan, and~ C. Emre ~Koksal}

\maketitle
\begin{abstract}
We consider the problem of link scheduling for wireless networks
with fading channels, where the link rates are varying with time. Due to the high
 computational complexity of the throughput optimal scheduler, we provide a low
 complexity greedy link scheduler GFS, with provable performance guarantees. We show that the performance of our greedy scheduler can
  be analyzed using the Local Pooling Factor (LPF) of  a network graph, which has been
   previously used to characterize the stability of the Greedy Maximal Scheduling (GMS) policy for networks with static channels. We conjecture that the performance of GFS is a lower bound on the performance of GMS for wireless networks with fading channels.
\end{abstract}
\vspace{-0.1in}
\section{Introduction}

The link scheduling problem for wireless networks has received
considerable attention in the recent past. In a wireless network
with shared spectrum, interference from neighboring nodes prevents
all nodes in the network from transmitting simultaneously at full
interference free rate. A link scheduler chooses a set of links to
deactivate at every time instant to eliminate their interference on
other links and only active links transmit data. An important
performance objective of a scheduler is throughput optimality,
\emph{i.e.,} for any given network, the scheduler should keep all
the queues in the network stable for the largest set of arrival
rates that are stabilizable for that network.\par For wireless
networks in which a set of link activation vectors are defined
according to a general binary interference model, the
Maxweight policy or the dynamic back-pressure policy is known to be
throughput optimal \cite{Tass}. Maxweight type policies have also been
shown to be throughput optimal for wireless networks with fading
channels, where the link rates vary over time \cite{Stolyar,Eryilmaz}.
However, the Maxweight policy suffers from high computational
complexity (NP-hard in many cases, including $k$-hop interference
models, k$>$1) \cite{sharma}, and has therefore motivated the study
of schedulers that have low complexity, are amenable to distributed
implementation and also offer provable performance guarantees.
Examples of such schedulers include Greedy Maximal Scheduling (GMS)\cite{Linshroff}
and Maximal Scheduling\cite{Chaporkar}, which have been widely studied for wireless
networks with static channels. \par 

There has been a number of studies that analyze the performance of GMS as a function of the network topology. The main parameter of focus has been efficiency, which is defined as the largest fraction of the network capacity region guaranteed to be stable under GMS. In \cite{Linshroff}, efficiency has been evaluated as a function of the local pooling factor of a network graph (LPF),
which depends on the network topology and interference constraints.
Later, using the LPF, GMS has been shown to be throughput optimal for a
wide class of network graphs under the node exclusive interference
model \cite{Birand, Zussman}.\par The performance analysis of the
aforementioned low complexity schedulers does not however, carry
over to the scenario with fading, in which link rates are
time-varying. For instance, unlike a static network, one cannot
conclude in a network with time-varying links that satisfying local
pooling under GMS implies throughput optimality. It is only known
that in the case of the node-exclusive interference model, GMS can
achieve at least half the network stability region. Thus, it is of interest to investigate if for networks with time varying link rates, GMS performs  as well as it does in networks with fixed link rates.
\cite{Linimperfect}. \par  In this paper, we develop a greedy link
scheduler, GFS, for wireless networks with fading channels, which,
although not throughput optimal, has low computational complexity
and offers provably good performance guarantees. We show that the
performance of our greedy scheduler can be related to the LPF of a
network graph. We then conjecture that the performance of GFS is a lower bound on the performance of GMS for wireless networks with time-varying link rates.
\vspace{-0.05in}
\section{System Model}
\begin{figure}
\psfrag{a1}{$S_1$\vspace{10mm}} \psfrag{a2}{\vspace{4mm} $S_2
\newline$} \psfrag{c11}{\hspace{-3mm} $c_{1}^{1}$}
\psfrag{c21}{\vspace{10mm}$c_{2}^{1}$} \psfrag{c12}{$c_{1}^{2}$}
\psfrag{c22}{$c_{2}^{2}$} \psfrag{c1b}{$\overline{c}_{1}$}
\psfrag{c2b}{$\overline{c}_{2}$} \psfrag{L}{$\Lambda$}
\psfrag{l1}{$l_1$} \psfrag{l2}{$l_2$}
\begin{center}
\includegraphics[bb=0bp 50bp 1165bp 380bp,clip,scale=0.18]{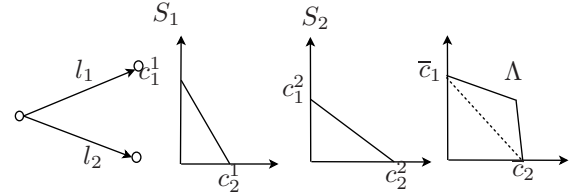}
\par
\end{center}
\caption{Figure shows an example of two interfering links with two fading states $S_1$ and $S_2$, occurring with probability $\pi^1$ and $\pi^2$. The network stability region region, $\Lambda$ is the interior of the region enclosed by the solid lines.}
\label{fig:capacity-region}
\vspace{-0.2in}
\end{figure}
We consider a
wireless network modeled as a graph
$\mathcal{G}=(\mathcal{V},\mathcal{E})$ with edges representing
links. We assume a single hop traffic model where each edge
represents a source-destination pair. Time is divided into slots and
packets arrive at the source node following an i.i.d. process with
a finite mean at the start of each time slot. Let $A_l(t)$ denote the number of packets arriving during time slot $t$. $A_l(t)$ has a mean $\lambda_l$. The vector of channel
states across all links in the network is assumed to be fixed over
the duration of a time slot but changing after every time slot. The
set of channels in the network can assume a state $j \in \{1,\ldots
,J \}$ according to stationary probability $\pi^j$. In each time
slot $t$, the achievable rate of link $l\in\mathcal{E}$, denoted by
$c_l(t)$, assumes value $c_{l}^{j}$ if the network is in fading state
$j$ at time slot $t$. The expected rate of a link, denoted by
$\overline{c}_{l}$ is given by $\overline{c}_l=\sum_{j=1}^{J}\pi^j
c_{l}^{j}$. We assume a generalized binary interference model, in
which each link $l$ is associated with an interference set, denoted
by $\mathcal{I}_l\subset\mathcal{E}$. Set $\mathcal{I}_l$ consists
of the set of links that cannot be active whenever link $l$ is
active.\par 
Let $\vec{r}^{j}$ denote a $1\times|\mathcal{E}|$
rate allocation vector for a network that is in channel state $j$,
where $r^{j}_{i}$ is the rate allocated to link $i\in\mathcal{E}$.
Any rate allocation vector $\vec{r}^{j}$ must satisfy the following
properties:
\begin{itemize}
\item[{(a)}]$r^{j}_{l}>0$ implies $r^{j}_{k}=0$, for all
$k\in\mathcal{I}_l$, where $k\neq l$
\item[{(b)}]There exists no link $k\in \mathcal{E}$ such that
$r^{j}_{k}\neq c^{j}_{k}$ and $k\notin\mathcal{I}_l$ for all $l$
satisfying $r^{j}_{l}>0$. In other words there exists no link that
does not interfere with any other active link and is yet not
scheduled.
\end{itemize}
Let $\mathcal{R}^j$ denote the set of all feasible rate allocation
vectors for a wireless network graph when the network is in channel 
state $j$. Similarly, $\mathcal{R}^j_{\mathcal{L}}$ is the set of all 
feasible rate allocation vectors on the subgraph
$\mathcal{L}\subset\mathcal{G}$. The stability region of the
network~\cite{Tass}\cite{NeelyTass} is then given by the interior of the set
$\Lambda=\{\vec{\lambda}:\vec{\lambda}\preceq\vec{\phi},\ \text{for
some}\ \vec{\phi}=\sum_{j\in{1,\ldots,J}}\pi^j\psi^{j} \}$, 
where $\psi^j \in Co(\mathcal{R}^j)$, with $Co(\mathcal{R}^j)$ representing the 
convex hull of the set $\mathcal{R}^j$, and $\preceq$
denoting component-wise inequality. Fig.~\ref{fig:capacity-region} shows an example of a simple two link
network with two fading states, with the associated network
stability region under a node-exclusive interference model. Figs.~\ref{fig:capacity-region}(a) and Figs.~\ref{fig:capacity-region}(b) illustrate the achievable rate regions in state $S_1$ and $S_2$ respectively. The network stability region $\Lambda$ is shown in Fig.~\ref{fig:capacity-region}(c).\par In related work,
\cite{Stolyar} considered a queueing model analogous to a cellular
network with $N$ links, where the network channel state followed an
irreducible discrete time Markov chain with a finite state space. It was
shown
that the policy which selects the queue with the highest
weight.
\[\max_{l=1\ldots N} q_{l}^{\beta}(t)c_{l}(t)\]
in each time slot, where $q_l$ is the queue size for link $l$  was throughput optimal for this network. In \cite{Neely}, it was shown that a
Maxweight-type scheduling policy was throughput optimal for power
allocation in wireless networks with time varying channels.
Similarly, \cite{Eryilmaz} also showed throughput optimality of a class of
Maxweight type policies for wireless networks with fading channels. \par 
Before we describe our greedy scheduler, we discuss the performance of non-opportunistic schedulers in the following section. In particular, we focus on a scheduler that utilizes only the mean link
rates, instead of instantaneous link rates. For this scheduler, we illustrate that when arrivals are correlated with channel states, 
the non-opportunistic policy can potentially keep serving links that
are experiencing  poor channel states, leading to a loss in throughput.

\vspace{-0.08in}
\subsection{Performance of Non-opportunistic Schedulers} 
We show
that a scheduler that utilizes the mean link
rates, instead of instantaneous link rates could perform arbitrarily
worse in certain cases. To illustrate this, we consider the 
two link network graph shown in  Fig. \ref{fig:capacity-region}. 
In this example, each link $l$ has one queue $Q_l$, into which packets
arrive according to an IID process.
Suppose that the rates of the two links
in each of the channel states are given by $c_{1}^{1}=1,c_{2}^{1}=\epsilon$, and 
$c_{1}^{2}=\epsilon,c_{2}^{2}=1$ respectively. Also, let $\pi^1$ be the 
probability with which the network assumes channel state 1, and 
$\pi^2$ be the probability for network channel state 2. In each time slot, 
the greedy non-opportunistic scheduler that we consider serves the 
queue which maximizes the quantity $Q_l(t)\bar{c_{l}}$.  
We will now construct an arrival traffic for this network under which
the queues for both links grow unbounded under the
non-opportunistic scheduling scheme. \par 
Let the initial queue lengths be $Q_1(0)=Q_2(0)=0$. 
At the beginning of each time slot, packets arrive according to the following statistics:
\begin{itemize}
\item [{(i)}]If the network channel state is 1, then with probability  
$1-\delta$, for an arbitrary $\delta>0$, $\epsilon$ packets arrive into the queue $Q_{1}$, and none 
into $Q_2$; With probability
$\delta$, $C/\bar{c_{1}}+\epsilon$ packets arrive into the queue of link 1,
and $C/\bar{c_{2}}$ packets arrive into the queue of link 2 respectively.  
$C$ is a fixed positive quantity.
\item [{(ii)}]If the network channel state is 2, then with probability  
$1-\delta$, $\epsilon$ packets arrive into the queue $Q_2$, and none into
$Q_1$; With probability
$\delta$, $C/\bar{c_{1}}$ packets arrive into the queue of link 1, and
$C/\bar{c_{2}}+\epsilon$ packets arrive into the queue of link 2 respectively.
\end{itemize}

Under this arrival statistic, we show that the end of each time slot, the length of each queue either remains unchanged or increases by a fixed quantity $C/\bar{c_{i}}$. 
At the beginning of the first time slot, all queues are assumed to be empty. The non-opportunistic scheduler then serves the queue with the highest weight, 
\emph{i.e.,} the queue into which $\epsilon$ or $C/\bar{c_{i}}+\epsilon$ packets have arrived. At the end of each time slot, the queue lengths remain unchanged with probability $1-\delta$, or increase by a fixed quantity $C/\bar{c_{i}}$ with probability $\delta$. Moreover, the queue lengths are also equal at the end of each time slot and of the form $kC$, where $k$ is a nonnegative 
integer. Since the queue length process
is non-decreasing, and the event that the queue length increases by a 
fixed positive quantity occurs infinitely often, the network is unstable
under the greedy non-opportunistic scheduler. The arrival rate vector of our
proposed arrival traffic is determined as
 $\vec{\lambda}=\pi^1(1-\delta)[\epsilon\quad 0]
 +  \pi^1\delta[{\scriptstyle C/\bar{c_1}}\qquad {\scriptstyle C/\bar{c_{2}}}]  +
\pi^2(1-\delta)[0\qquad\epsilon]
 +\pi^2\delta[{\scriptstyle C/\bar{c_1}} \qquad {\scriptstyle C/\bar{c_{2}}}]$, 
which simplifies to $\vec{\lambda}=\epsilon\left[\pi^1 \qquad \pi2] + 
\delta[\frac{{\scriptstyle C}(\pi^1+\pi^2)}{\bar{c_1}}-\epsilon\pi^1 \qquad 
\frac{{\scriptstyle C}(\pi^1+\pi^2)}{\bar{c_2}}-\epsilon\pi^2\right]$. Thus,
when $\epsilon$ is small, the greedy non-opportunistic scheduler is unable to support 
arrival rates that are within a fraction $\epsilon$ 
of the stability region. Note that in the above example, the arrival process
is correlated with the network channel state process. 
\vspace{-0.1in}
\section{ A Greedy scheduler for Networks
with Fading Channels (GFS)} \label{sec:GFS}
 The greedy scheduler that we propose is similar to GMS except that it requires each link to have
a virtual queue corresponding to every channel state of the
network, \emph{i.e.,} each link has a set of $J$ virtual queues. In each time slot, packets arriving into a link $l$ are placed into one of the $J$ queues.  In practice, each link could maintain only one real first-in first-out queue, into which packets arrive and depart, and counters for the virtual queues which keep track of the number of packets in the virtual queue. The GFS scheduler would use the values of the counters to make the scheduling decision. Using such counters, also known as shadow queues have been effective in reducing queueing complexity and delay \cite{Bui}. Let $\mathrm{Q}_{l}^{j}$ be the
virtual queue of link $l$ corresponding to fading state $j$ and $\mathit{q}_{l}^{j}(t)$
denote its size at time $t$. Let $\mathrm{Q}_{l}$ denote the real FIFO queue of link  $l$. 
We now
describe our greedy scheduler:
\begin{itemize}
\item[{(1)}] At the beginning of time slot
$t$, packet arrivals $A_l(t-1)$ are placed in queue $\mathrm{Q}_{l}^{j}$ with
probability $\frac{\pi^jc_{l}^{j}}{\overline{c}_l}$.
\item[{(2)}] In time slot $t$, let the network be in fading state $j$.
GFS observes only the queues
corresponding to fading state $j$, in order to select the rate
allocation vector. The scheduler first selects the link with highest
weight $m=\argmax_{l\in \mathcal{E}}\mathit{q}_{l}^{j}c_{l}^{j},$ removes all links in $\mathcal{I}_m$ from the set
of potential links to be scheduled at time $t$, and repeats the
process until there are no more non-interfering links that remain to
be selected. 
\end{itemize}
At the end of this procedure GFS selects a rate allocation vector
that belongs to $\mathcal{R}^j$, when the network channel state is $j$. Note that the GFS policy becomes identical to GMS in the case of networks with static link rates. Also, the application of the GFS
policy on the queues corresponding to fading state $j$, requires the
knowledge of the network fading state at every node in the network.  The departure process for the virtual queues can now be described as follows: For any link $l$, $\min(\vec{r}(l),Q_{l}^{j}(t))$ packets depart from virtual queue $q_{l}^{j}$, while $\min(\vec{r}(i),Q_{l}(t))$ packets depart from the real FIFO queue $q_l$. 
\vspace{-0.1in}
\subsection{Performance Analysis of GFS}
We now give the
main result of this paper, which uses the LPF of a
network graph to evaluate the stability region achievable using GFS. Before we state our result, we define the following static
wireless network: given any wireless network graph
$\mathcal{G}=(\mathcal{V},\mathcal{E})$ with time varying link
rates, we associate with $\mathcal{G}$ a static wireless network
$\hat{G}=(\mathcal{V},\mathcal{E})$, whose link rates are fixed at
$\bar{c}_{l},\, \forall l$. Let $\hat{\mathcal{R}}$ denote the set of 
all feasible rate allocation vectors for the network graph $\hat{G}$. We also define
$\mathcal{G}^j=(\mathcal{V},\mathcal{E})$ to be a static network 
whose link rates are fixed at $c^{j}_{l},\, \forall l,\, j=1,\ldots,J$. 
Finally, we let $\Lambda$ and $\widehat{\Lambda}$ denote the network stability regions of the
networks $\mathcal{G}$ and $\hat{\mathcal{G}}$ respectively. Note
that $\widehat{\Lambda}\subseteq\Lambda$, since 
The LPF for 
the network graph $\hat{\mathcal{G}}$ can then be defined as follows \cite{JooShroff}:
\begin{defn}
\label{def:def1}Let $\mathcal{L}$ be any subgraph of $\hat{\mathcal{G}}$. Then
$\mathcal{L}$ satisfies $\sigma$-local pooling if, for any given pair  $\vec{\mu},\vec{\nu}$, where $\vec{\mu}$ and $\vec{\nu}$ are convex combinations
of the rate vectors in $\hat{\mathcal{R}}_{\mathcal{L}}$, 
we have $\sigma\vec{\mu}\nprec\vec{\nu}$. \\
The LPF  $\sigma^{*}$, for the network is then defined as:
\begin{align*} 
\sigma^{*}=\sup\left\{\sigma\mid\forall\, \mathcal{L}\subset
\hat{\mathcal{G}},
\mathcal{L} \text{ satisfies } \sigma\text{-local pooling} \right\}.
\end{align*}
\end{defn}\vspace{-0.06in}

The LPF of a network graph depends only on the topology of the network graph and therefore
is identical for $\hat{\mathcal{G}}$ and 
$\mathcal{G}^j, \,j\in\{1,\ldots,J\}$.

\begin{theorem}\label{th1}
Let $\sigma^{*}$ be the LPF of a network $\hat{\mathcal{G}}$. Then,
the network $\mathcal{G}$ is stable under the GFS policy for all arrival rate
vectors $\vec{\lambda}$ satisfying $\vec{\lambda}\in\sigma^{*}
\hat{\Lambda}$, where $\hat{\Lambda}$ is the stability region of the
corresponding network graph $\widehat{\mathcal{G}}$.
\end{theorem}
Theorem 1 provides performance guarantees for our scheduling policy for any wireless network in terms of the
 stability region of an associated identical static network whose link rates are fixed at their expected rates. Note that an LPF of 1 implies that the associated greedy policy can guarantee stability for any arrival rate in $\hat{\Lambda}$. 
Examples of network graphs which have $LPF=1$ include tree network
graphs under the $k$-hop interference model for $k\geq1$. In
\cite{Birand}, all network graphs with $LPF=1$ under the
node-exclusive interference model are identified.\par
We prove Theorem 1 by first establishing the stability of the virtual queues. We then provide Lemma \ref{lemma3} to establish stability of the real FIFO queues as well.
\begin{IEEEproof}
We consider the fluid limit model of the system. Let
$\vec{A_{l}^{j}(t)}$ denote the cumulative arrival process into
queue $\mathit{q}_{l}^{j}$ and $S_{l}^{j}(t)$ denote the cumulative
service process for $\mathrm{Q}_{l}^{j}$ until time slot $t$. For the
arrival and service processes, we use $A_{l}^{j}(t)=A_{l}(\lfloor
t\rfloor)$, and $S_{l}^{j}(t)=S_{l}^{j}(\lfloor
t\rfloor).$ For the queue process $\mathit{q}_{l}^{j}(t)$, we employ
linear interpolation.\par We now consider a sequence of scaled queuing
systems $(\vec{\mathit{q}}^{n}(\cdot),\vec{A^{n}}(\cdot),\vec{S^{n}}(\cdot))$.
where we apply the scaling ${\mathit{q}_{l}^{j}(nt)/n},\;
{A_{l}^{j}(nt)/n},\mbox{ and } S_{l}^{j}(t)(nt)/n,\,\forall
l\in\mathcal{E}$ with the queue process satisfying
$\sum_{l\in\mathcal{E}}{\mathit{q}_{l}^{j}}(0)\leq n$. Then, using the
techniques to establish fluid limit in \cite{Dai}, one can show that
a fluid limit exists almost surely, \emph{i.e,} for almost all
sample paths and for any positive $n\rightarrow\infty$, there exists
a sub-sequence $n_{k}$ with $n_{k}\rightarrow\infty$ such that
following convergence holds uniformly over compact sets: For all
$l\in\mathcal{E},$
$\frac{1}{n_{k}}\{A_{l}^{j}\}^{n_{k}}(n_{k}t)\rightarrow\frac{\pi^jc_{l}^{j}}{\overline{c}_{l}}\lambda_{l}t,\,j\in1\ldots
J\,$,
$\frac{1}{n_{k}}\{S_{l}^{j}(t)^{n_{k}}(n_{k}t)\rightarrow
S_{l}^{j}(t)$, and
$\frac{1}{n_{k}}\{\mathit{q}_{l}^{j}\}^{n_{k}}(n_{k}t)\rightarrow
\tilde{\mathit{q}}_{l}^{j}(t)$, where $\tilde{\mathit{q}}_{l}^{j}(t)(t)$ and
$S_{l}^{j}(t)$ are the fluid limits for the queue length
processes and the service rate processes respectively. The fluid
limit is absolutely continuous and hence the derivative of
$\tilde{\mathit{q}}_{l}^{j}(t)$ exists almost everywhere \cite{Dai}
satisfying:
\begin{equation}
{\frac{d}{dt}\tilde{\mathit{q}}_{l}^{j}(t)\ =\begin{cases}
\left[\frac{\pi^jc_{l}^{j}}{\overline{c}_l}\lambda_{l}
-\gamma_{l}^{j}(t)\right]^{+} & \mathit{q}_{l}^{j}(t)>0\\
0 & \mbox{otherwise}\end{cases}}\label{diffeq}
\end{equation}

where $\gamma^{j}_{l}(t)=\frac{d}{dt}(s_{l}^{j}(t))$.

Consider the times $t$ when the derivative $\frac{d}{dt}\mathit{q}^{j}_{l}(t)$
exists for all $l\in\mathcal{E},j\in{1,\ldots, J}$. Let $L_{0}(t)$ denote the set of
queues with the largest weight, \emph{i.e.,} $L_{0}(t)=\argmax_{\mathrm{Q}_{l}^{j}\in\Psi}
\tilde{\mathit{q}}_{l}^{j}(t)c_{l}^{j},$
where $\Psi$ is the set of all queues in the network. Let $L(t)$
denote the set of queues from $L_{0}(t)$, which have the maximum
derivative of the weights, \emph{i.e.,} $L(t)=\argmax_{\mathrm{Q}_{l}^{j}\in L_{0}(t)}\frac{d}{dt}
\tilde{\mathit{q}}_{l}^{j}(t)c_{l}^{j}.$
The set $L(t)$ can then be expressed as
$\underset{j}{\bigcup}L^j(t)$, where \vspace{-0.1in}
\[L^j(t)=\{\mathrm{Q}_{l}^{j},l \in \mathcal{E}\mid 
\mathrm{Q}_{l}^{j}\in L(t)\},\,j=1\ldots J.\]

Since $\tilde{\mathit{q}}_{l}^{j}(t)$ is absolutely continuous, there
exists a small $\delta>0$ such that in the interval $(t,t+\delta)$, the weight of queues in $L^j(t)$ dominates
the weight of other queues, whenever the network channel state is $j$. Hence, GFS gives
priority to queues belonging to $L^j(t)$ in 
$(t,t+\delta)$. 
We now provide the following two lemmas to characterize the arrival
rates and service rates for the queues in $L^j(t)$. Let 
$\mathcal{E}_{L^j(t)}\subset\mathcal{G}^j$ denote the set of links whose queues are in $L^j(t)$.
Thus $\mathcal{R}^{j}_{\mathcal{E}_{L^j(t)}}$ denotes the set of all feasible
rate allocation vectors for the subgraph $\mathcal{E}_{L^j(t)}$. 
Let $\vec{\lambda}^{j}$ be the $|\mathcal{E}|$ dimensional arrival rate vector whose each element $\lambda^{j}(l)$ represents the arrival rate into queue $\mathrm{q}_{l}^{j}$. For any $|\mathcal{E}|$ vector $\vec{\eta}$, the projection of 
$\vec{\eta}$ on a subset of edges $L$, denoted by $\vec{\eta}|_{L}$, is defined as 
the $|L|$ dimensional vector obtained by restricting $\vec{\eta}$ to the edges in 
$L$.

\begin{lem} \label{lemma1}
Consider any fading state $j \in {1\ldots J}$ such that
 $L^j(t)\neq \emptyset$. If the arrival rate vector $\vec{\lambda}\in\sigma^{*}\hat{\Lambda}$, then $\vec{\lambda}^{j}$, the arrival rate into the queues
 $\mathrm{q}_{l}^{j} \,\forall l\in\mathcal{E}$, when  projected on to the set of links 
 $\mathcal{E}_{L_{j}(t)}$ can be expressed as $\vec{\lambda}^{j}|_{\mathcal{E}_{L^j(t)}}=\sigma^{*}\pi^j\vec{\mu},$
where $\vec{\mu}$ is a convex combination of the rate allocation vectors in
 $\mathcal{R}^{j}_{\mathcal{E}_{L^j(t)}}.$  
\end{lem}
\begin{proof}
Since $\vec{\lambda}\in\sigma^{*}\hat{\Lambda}$, it satisfies $\vec{\lambda}\preceq\sigma^{*}\vec{\Phi}$, for some  $\vec{\Phi}=\sum_{i}\alpha_{i}\vec{\hat{r}}_{i}$, where $\vec{\hat{r}}_{i}\in\hat{\mathcal{R}}$, and $\sum_i\alpha_i=1$. One can then write the arrival rate into a link $l$ as $\lambda_{l}=\sigma^{*}\overline{c}_{l}\sum_{i}\alpha_i\mathbf{1}_{\{\hat{r}_{i}(l)\neq0\}},$
where $\mathbf{1}_{\{\hat{r}_{i}(l)\neq0\}}$ is the indicator function. The arrival rate into queue $\mathrm{Q}_{l}^{j}$ is then given by
$\lambda_{l}^{j}=\sigma^{*}\frac{\pi^jc_{l}^{j}}{\overline{c}_{l}}\overline{c}_{l}\sum_{i}\alpha_i\mathbf{1}_{\{\hat{r}_{i}(l)\neq0\}},$
which yields
$\lambda_{l}^{j}=\sigma^{*}\pi^{j}c_{l}^{j}\sum_{i}\alpha_i\mathbf{1}_{\{\hat{r}_{i}(l)\neq0\}}$, for all $l\in\mathcal{E},\text{and } j\in\{1\ldots J\}.$
We can then write the arrival rate vector $\vec{\lambda}^{j}$ in terms of rate allocation vectors  in $\mathcal{R}^{j}$ as $\vec{\lambda}^{j}=\sigma^{*}\pi^{j}\sum_{i}\alpha_i \vec{r}^{j}_{i},$
since if $\mathbf{1}_{\{\hat{r}_{i}(l)\neq0\}}=1, \text{ then } 
\mathbf{1}_{\{r^{j}_{i}(l)\neq0\}}=1$, or $c_{l}^{j}=0$. It follows that $\vec{\lambda}^{j}|_{\mathcal{E}_{L^j(t)}}
\preceq\sigma^{*}\pi^j\vec{\mu}$, where $\vec{\mu}$ is a convex combination of the rate allocation vectors in $\mathcal{R}^{j}_{\mathcal{E}_{L^j(t)}}.$  
\end{proof}

\begin{figure}
\subfloat[\label{fig:four-link graph} ]{\noindent \begin{centering}
\includegraphics[bb=20bp 0bp 400bp 400bp,scale=0.3]{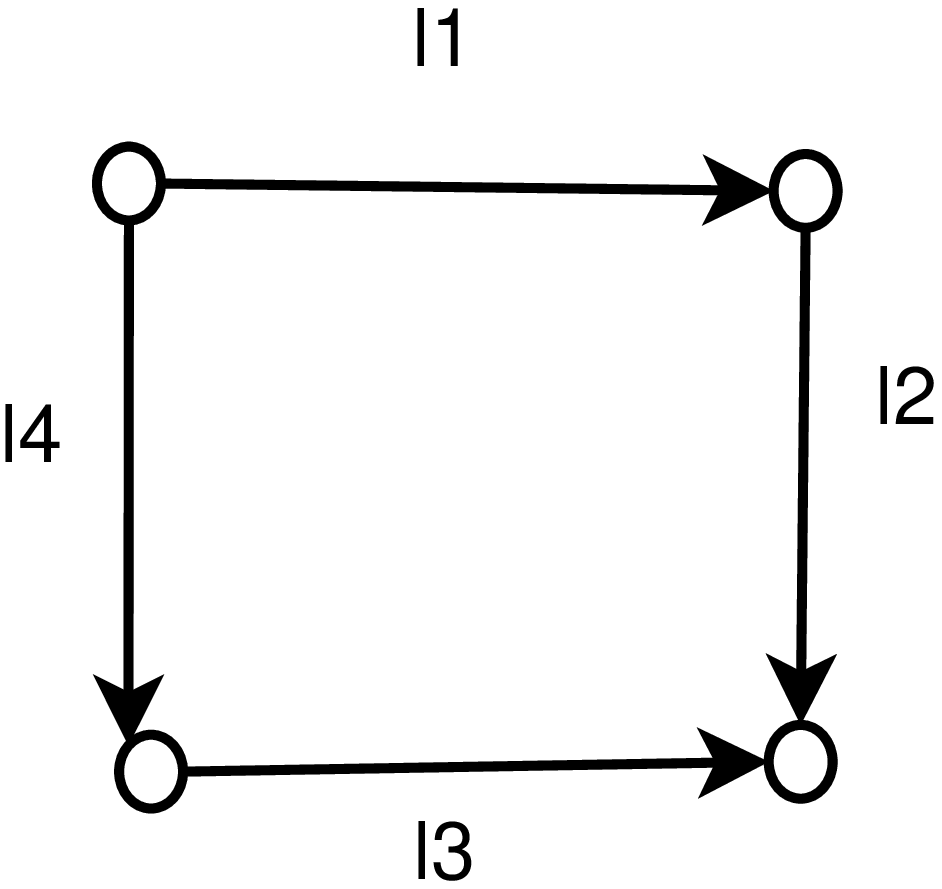}
\par\end{centering}

\centering{}}\subfloat[\label{fig:GFScomparison} ]{\noindent \begin{centering} 
\includegraphics[bb=20bp 170bp 300bp 500bp,scale=0.3]{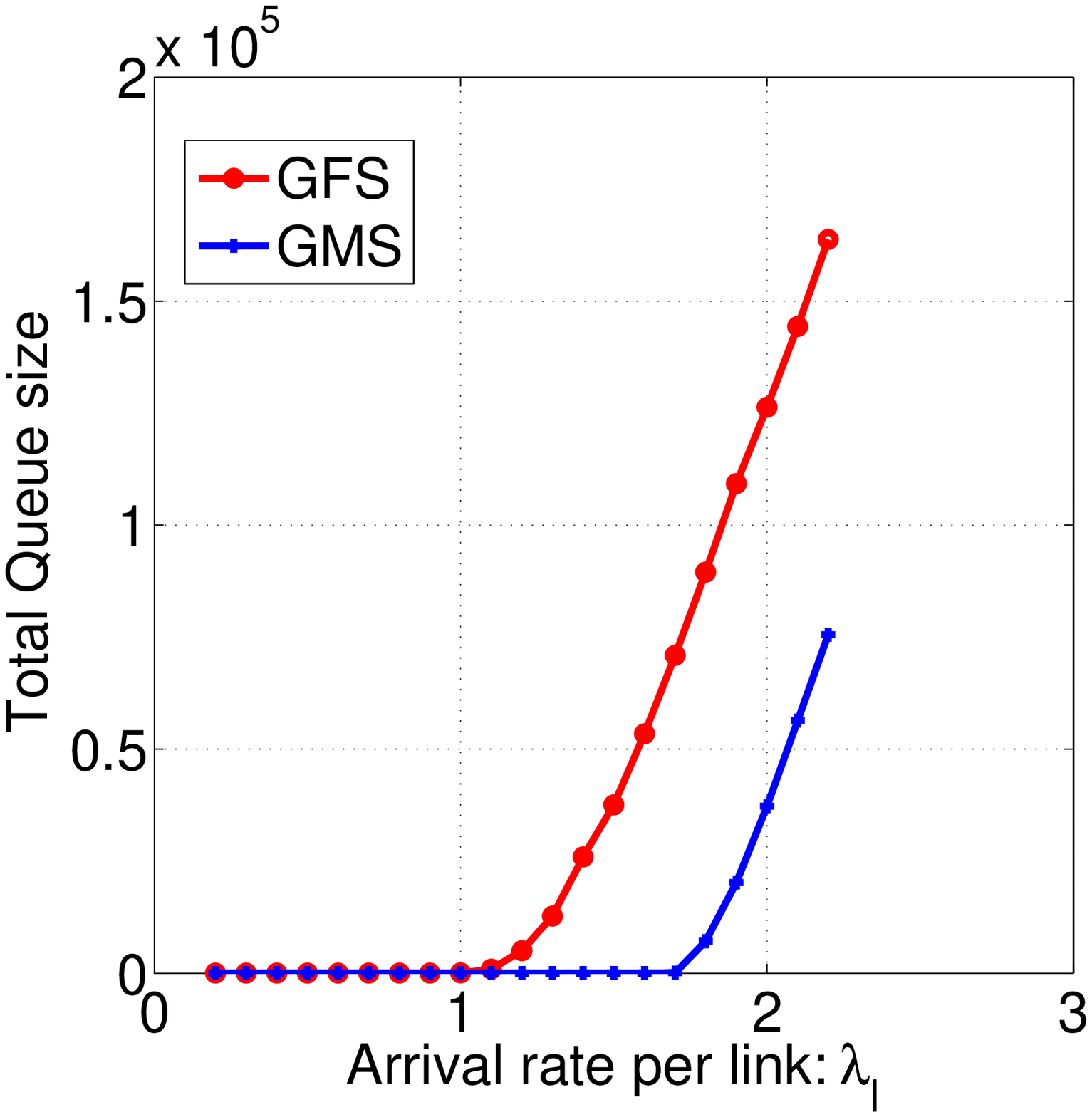}
\par\end{centering}}
\centering{}
\noindent \caption{A four-link network graph is shown in 
Fig. ~\ref{fig:four-link graph} and the performance of GFS and GMS is plotted in Fig. ~\ref{fig:GFScomparison}.}
\vspace{-0.25in}
\centering{}
\end{figure}

\begin{lem} \label{lemma2}
Consider any fading state $j \in {1\ldots J}$ such that $L^j(t)\neq \emptyset$. Then the  service rate vector $\vec{\gamma}^{j}(t)$, projected onto the links in
$\mathcal{E}_{L_{j}(t)}$, can be expressed as $\vec{\gamma}^{j}|_{\mathcal{E}_{L^j(t)}}=\pi^j\vec{\nu},$
where $\vec{\nu}$ is a convex combination of the rate allocation vectors in
 $\mathcal{R}^{j}_{\mathcal{E}_{L^j(t)}}.$
\end{lem}  
\begin{proof}
The full proof is similar to that in  
 \cite{Stolyar} and \cite{JooShroff} and is omitted here. 
Consider all queues belonging to $L^j(t)$. Since the queues in $L^j(t)$ have 
the highest weight in $(t,t+\delta)$ when the network is in state $j$, the GFS scheduler gives priority
queues in $L^j(t)$ whenever the network channel state enters state $j$ in the time interval $(t,t+\delta)$. Consequently, the rate allocation vectors selected by GFS
in network channel state $j$, when projected on the set of links $\mathcal{E}_{L^j(t)}$
yields an element from the set
$\mathcal{R}^{j}_{L^j(t)}$. Therefore, under the GFS policy, 
the service rate vector  $\vec{\gamma}^{j}(t)$ for the set of queues $\vec{\mathrm{Q}}^j$, 
projected onto $\mathcal{E}_{L^{j}(t)}$ is a convex combination of the elements 
of $\mathcal{R}^{j}_{\mathcal{E}_{L^j(t)}}$. From the ergodicity of the network channel
state process, GFS serves elements in $\mathcal{R}^{j}_{L^j(t)}$ a fraction
$\pi^j$ of the time. It follows that $\vec{\gamma}^{j}|_{L_{j}(t)}=\pi^j\vec{\nu}$. \end{proof}
From Lemma \ref{lemma1} and Lemma \ref{lemma2}, 
the arrival rate $\vec{\lambda}^{j}|_{\mathcal{E}_{L^j(t)}}$ as well as
the service rate $\vec{\gamma}^{j}|_{\mathcal{E}_{L^j(t)}}$ can be 
expressed in terms of the convex combinations of elements in 
$\mathcal{R}^{j}_{L^j(t)}$. Since $\hat{\mathcal{G}}$ satifies $\sigma^{*}$ local pooling, 
$\mathcal{E}_{L^j(t)}$ being a subgraph of $\mathcal{G}^j$ satisfies
$\sigma^{*}$ local pooling. It follows from the definition of
 $\sigma$-local pooling that there exists a link 
$l\in\mathcal{E}_{L^j(t)}$ such that its queue 
$\mathrm{q}_{l}^{j}$ satisfies 
$\lambda_{l}^{j}-\gamma_{l}^{j}\leq-\epsilon$,  
for some $\epsilon >0$. Since 
$\frac{d}{dt}\tilde{\mathit{q}}_{l}^{j}(t)=\frac{d}{dt}
\tilde{\mathit{q}}_{m}^{n}(t)$ for any pair
$\mathrm{Q}_{l}^{j},\mathit{Q}_{m}^{n} \in L(t)$, we obtain 
$\frac{d}{dt}\tilde{\mathit{q}}_{l}^{j}(t)<\epsilon$,
for all $\mathrm{Q}_{l}^{j}\in L(t)$. \par 
We now consider the Lyapunov function $V(t)=\max_{l\in \mathcal{E},j\in \{1\ldots J\}}
\tilde{\mathit{q}}_{l}^{j}.$
The derivative of $V(t)$ is given by :
\begin{align*}
\frac{d}{dt}V(t)&=\frac{d}{dt}\max_{l\in\mathcal{E},j\in\{1\ldots J\}}\tilde{\mathit{q}}_{l}^{j}c_{l}^{j}
\leq\max_{\tilde{\mathit{q}}_{l}^{j}\in L(t)}\frac{d}{dt}
\tilde{\mathit{q}}_{l}^{j}c_{l}^{j}\leq-\epsilon.
\end{align*}
The negative drift of the Lyapunov function implies that the 
fluid limit model of the system is stable and hence by Theorem 4.2 
 in \cite{Dai}, the original
system is also stable.


\begin{lem} \label{lemma3}
Consider any sequence of arrivals $A_l(t), t=1,2,3\cdots,$ for all $l\in\mathcal{E}$. Then under the GFS policy, we have $q_l(t)\leq \sum_{j=1}^{J}q^{j}_{l}(t) + B,\, \forall t=1,2,3,\cdots,\text{ and } \forall \,l\in\mathcal{E}$, where $B$ is a bounded positive number.
\end{lem}
\begin{proof}
Without loss of generality, we assume $B=0$. Suppose at the beginning of time slot $t, t\geq 0$, we have $q_l(t)\leq \sum_{j=1}^{J}q^{j}_{l}(t),\, \forall \,l\in\mathcal{E}$. Let $j$ denote the network state in time slot $t$. Then, if $D_l(t)$ and $D^{j}_l(t)$ denote the packets departing in time slot $t$ from the real FIFO queue $q_l$ and the virtual queue $q^{j}_l$ respectively , the following must be true:
If $D_l(t)=D^{j}_l(t)$, then in time slot $t+1$, we have $q_l(t+1)\leq \sum_{j=1}^{J}q^{j}_{l}(t+1)$, since both $q_l(t)$ and $\sum_{j=1}^{J}q^{j}_{l}(t)$ are incremented by the same number of arrivals $A_l(t)$.  Similarly, if $D_l(t)>D^{j}_l(t)$, then it again implies that $q_l(t+1)< \sum_{j=1}^{J}q^{j}_{l}(t+1)$. Finally, if $D_l(t)<D^{j}_l(t)$, it implies that $q_l(t)< \vec{r}^j_l(t)$. Consequently, $q_l(t)$ empties and  $q_l(t+1)=A_l(t)\leq\sum_{j=1}^{J}q^{j}_{l}(t+1)$.
Since  $q_l(t)\leq \sum_{j=1}^{J}q^{j}_{l}(t)$ is satisfied at $t=0$, we obtain the desired condition at any time $t$.
\end{proof}
Lemma 3 shows  that if the virtual queues are stable then the corresponding real FIFO queue is also stable.
\end{IEEEproof}
\section{Simulation}

In this section we simulate the performance of GFS for the four link network graph shown in Fig.~\ref{fig:four-link graph}. Each link independently assumes one of four different states in each time slot, where the link states correspond to rates 1, 2 , 3 and 4 units per time slot. The probability distribution of the link states are independent and non-identical across links, with the average link rates being $\overline{c}_1=2.7,\, \overline{c}_2=2.1, \overline{c}_3=2.8,\, \text{ and }  \overline{c}_4=3.1$ respectively. 
In Fig.~\ref{fig:GFScomparison}, we plot the total queue sizes as we uniformly increase the arrival rate into all links. The plots show that GFS is able to sustain a load of atleast 1 unit per link. Since the network in Fig.~\ref{fig:four-link graph} has LPF value of 1, GFS can stabilize the region $\hat{\Lambda}$. GFS therefore guarantees a per-link symmetric rate of at least 1, since the arrival rate $[1\; 1\; 1\; 1]$ lies inside $\hat{\Lambda}$. While the performance of GMS is better than GFS in the plot of Fig.~\ref{fig:GFScomparison}, the current known performance guarantee of GMS is only half the network stability region $\Lambda$ under the one-hop interference model, which corresponds to a symmetric load of 0.5 per link. Based on simulations, we conjecture that the performance of GFS is a lower bound on the performance of Greedy Maximal Scheduling in time varying wireless networks. The performance guarantees for GFS thus motivates the analysis of GMS for time varying networks as our future work.

\section{conclusion} 
We develop a greedy scheduler, GFS, for wireless networks with
time varying channel states and provide provable performance guarantees for 
this scheduler. Our greedy scheduler, though suboptimal, has low 
computational complexity and performs better than 
non-opportunistic schedulers that do not exploit instantaneous channel 
state information. The performance guarantees, along with simulations, also paint an optimistic picture of the performance of GMS in wireless networks with fading channels, and we conjecture the stability region guaranteed under GFS for any wireless network to be a  lower  bound on the stability region of GMS.

\end{document}